\documentclass{sig-alternate-05-2015}

\usepackage{graphicx}
\usepackage{comment}
\usepackage[linesnumbered, lined, boxed,commentsnumbered]{algorithm2e}
\usepackage{hyperref}
\usepackage{bm}
\usepackage[english]{babel}

\newcommand{\dlg}{DALiuGE}

\newcommand{\pgt}{Physical Graph Template}

\newtheorem{theorem}{Theorem}

\begin{document}




\isbn{123-4567-24-567/08/06}



%
\title{Partitioning SKA Dataflows for Optimal Graph Execution}
%
%
%
%
%

\numberofauthors{3} 
%
\author{
%
%
\alignauthor
Chen Wu\\
       \affaddr{International Centre for Radio Astronomy Research}\\
       \affaddr{The University of Western Australia}\\
       \affaddr{Perth, Australia}\\
       \email{chen.wu@icrar.org}
\alignauthor
Andreas Wicenec\\
       \affaddr{International Centre for Radio Astronomy Research}\\
       \affaddr{The University of Western Australia}\\
       \affaddr{Perth, Australia}\\
       \email{andreas.wicenec@icrar.org}
\alignauthor 
Rodrigo Tobar\\
       \affaddr{International Centre for Radio Astronomy Research}\\
       \affaddr{The University of Western Australia}\\
       \affaddr{Perth, Australia}\\
       \email{rtobar@icrar.org}
}

\date{1 April 2018}

\maketitle
\begin{abstract}
Optimizing data-intensive workflow execution is essential to many modern scientific projects such as the Square Kilometre Array (SKA), which will be the largest radio telescope in the world, collecting terabytes of data per second for the next few decades. At the core of the SKA Science Data Processor is the graph execution engine, scheduling tens of thousands of algorithmic components to ingest and transform millions of parallel data chunks in order to solve a series of large-scale inverse problems within the power budget. To tackle this challenge, we have developed the Data Activated Liu Graph Engine (DALiuGE) to manage data processing pipelines for several SKA pathfinder projects. In this paper, we discuss the DALiuGE graph scheduling sub-system. By extending previous studies on graph scheduling and partitioning, we lay the foundation on which we can develop polynomial time optimization methods that minimize both workflow execution time and resource footprint while satisfying resource constraints imposed by individual algorithms. We show preliminary results obtained from three radio astronomy data pipelines.
\end{abstract}

%
%


%
%

%
%



\section{Introduction}
The Square Kilometre Array (SKA) will be the largest radio telescope in the world~\cite{braun2015advancing}. The two components of the first phase of SKA (SKA1) --- SKA-Mid and SKA-Low --- will jointly produce large amounts of data at a rate of one Terabyte (TB) per second, with the second phase data rate reaching at least ten times higher. All this data has to be captured, reduced, processed, and analyzed in near real-time. This poses a great challenge, since the current generation of radio astronomy data processing systems are designed to handle data approximately two to three orders of magnitude smaller than that of the SKA1.

To tackle this challenge, we developed the Data Activated Liu Graph Engine (DALiuGE\footnote{\texttt{https://github.com/ICRAR/daliuge}}) to execute continuous, time-critical, data-intensive workflows in order to produce science-ready data products. Compared to existing astronomical workflow systems, DALiuGE has several advantages such as separation of concerns, data-centric execution, graph-based dataflow scheduling, and native support for streaming processing.

A technical overview of \dlg{} and its operational production systems are described in~\cite{wu2017daliuge}. In this paper, we focus on the \dlg{} graph scheduling sub-system. In particular, we discuss technical details on dataflow partitioning algorithms and implementations.

\begin{figure*}
\centering
\includegraphics[scale=0.3]{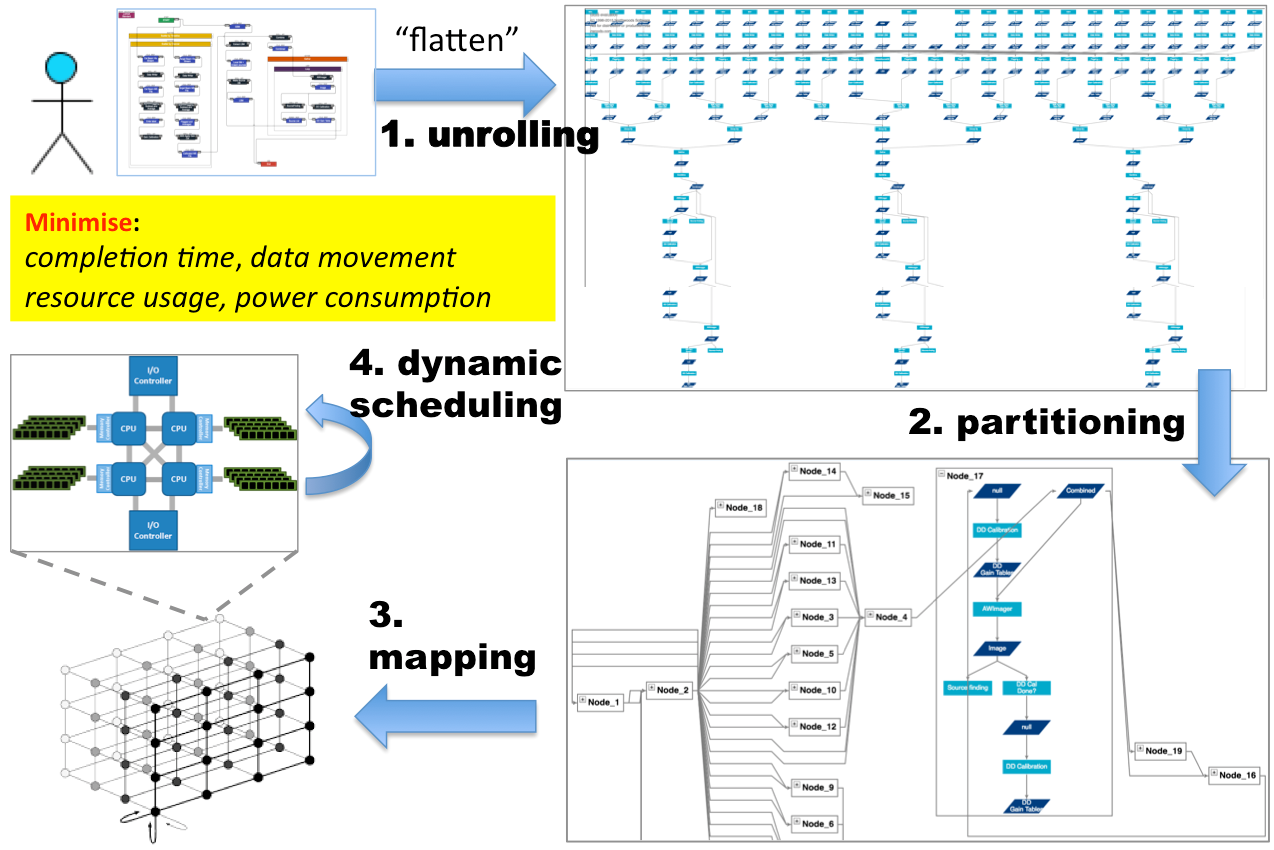}
\caption{The complete dataflow execution cycle consists of four major steps --- unrolling, partitioning, mapping and dynamic scheduling. Both unrolling and partitioning are performed offline. Mapping happens just a few minutes before the workflow execution, and dynamic scheduling is done in real-time during execution. While the first three steps target the entire graph across multiple resources, the last step focuses on tasks using local resources on a single node.}
\label{fig:overall}
\end{figure*}

\section{Related work}
The dataflow computation model~\cite{dennis1975preliminary} represents workflows as Directed Acyclic Graphs (DAG), where vertices are stateless computational tasks (i.e. functions) and edges connect the output of one task with the input of another. Although the dataflow model exploits parallelism inherent in DAGs through data dependencies, mapping an irregular DAG onto hardware resources for optimal execution is an NP-hard problem~\cite{chaudhary1993generalized}. Early work attempted to derive data structures (e.g. assignment graph~\cite{bokhari1981shortest} or allocation graph~\cite{towsley1986allocating}) from the original DAG in order to perform tractable searching and optimisation algorithms (e.g. using the maximum flow solutions~\cite{stone1977multiprocessor}). While these algorithms were able to uncover an optimal solution in polynomial time, the growth rate of the assignment graph is \(O(N \times M)\), where \(N\) denotes the number of vertices in the original DAG and \(M\) denotes the number of available processors. Therefore, as the DAG size and resource pool grows substantially (e.g. from tens of tasks running on a laptop to millions of tasks running on thousands of processors), these exact optimisation methods quickly become intractable.

%
%
%
%
A variety of heuristics-based algorithms~\cite{kwok1999static} have been developed for scheduling DAGs on multiprocessors. These heuristics in general fall into two alternative approaches --- one-phase or two-phase. In the one-phase approach (e.g., the widely-used HEFT algorithm~\cite{topcuoglu2002performance}), DAG scheduling is performed by directly mapping a ranked list of workflow tasks to another ranked list of resource units (e.g. processors or nodes) based on some aggregated run-time workflow profiles and resource statistics. In contrast, the two-phase approach ~\cite{liou1997comparison, sarkar1987partitioning} first partitions the DAG into a number of clusters based on heuristics such as load balancing~\cite{karypis1998multilevelk}, minimal data movement, etc. In the second phase, these clusters are then mapped onto actual hardware resources for execution. We currently adopt the two-phase approach because the output from the first phase encodes a resource demand abstraction (RDA) from intrinsic properties of the DAG. The RDA becomes the input for resource mapping in the second phase. More importantly, the RDA provides a more accurate estimate of resource demand for future capacity planning and observation scheduling for the telescope manager. However, most two-phase algorithms were targeted to multiprocessors on a single compute node, where each workflow task consumes exactly one processor. Our workflows need to run across clusters of compute nodes, each consisting of multiple processors. More importantly, each workflow task inherently demands multiple yet different number of processors/cores and different amount of memories. Dealing with this kind of complexity in resource demand and multiplicity in resource capabilities is one of our contributions in this paper. Moreover, unlike most existing DAG scheduling/mapping algorithms, our partitioning algorithm aims to reduce the overall resource footprint given these complexities and constraints.

On the other hand, the advantage of the one-phase approach is its flexibility to incorporate run-time resource heterogeneity. We leave for our future work a thorough investigation and application of the one-phase approach to our DAG mapping problem. 

Although significant progress~\cite{bateni2017affinity, martella2017spinner, tsourakakis2014fennel} has been made recently to partition vary large graphs for various social network analysis and machine learning applications, direct application of these graph partitioning algorithms for dataflow partitioning often leads to sub-optimal solutions. This is because the DAG (or general graph) representation \(G\) of the dataflow does not encode the notion of workflow execution working set \(W_t\) - a small set of workflow tasks that are being executed at time \(t\). Only tasks in \(W_t\) consume resources, other tasks are either waiting for the completion of their ``upstream" tasks in \(W_t\) or have already completed their executions. Therefore, partitioning the entire graph \(G\) (e.g. in the order of millions of nodes) for subsequent resource mapping is (1) wasteful given that \(|W_t| \ll |G|\), and (2) ill-posed since \(W_t\) is time-dependent and is unknown at the time of graph partitioning.

\section{Overview of Graph Execution}
Following the two-phase approach, the four steps of the graph execution are illustrated in Figure \ref{fig:overall}. We briefly introduce them in this section. Readers are referred to for a detailed technical discussion on graph execution.

Starting from the top left corner, a staff astronomer composes a logical graph representing high-level data processing capabilities (e.g., ``Image deconvolution") using resource oblivious dataflow constructs and workflow task components. The \textbf{first step} unrolls the logical graph by expanding all parallel branches and loops, instantiating tasks in all branches and iterations and connecting them with directed edges as per the logical graph definition. The result of unrolling is the \pgt{} (PGT) shown in the top right corner. It should be noted that, unlike traditional dataflow graph representations, DALiuGE models data as well as tasks as graph vertices. From a workflow viewpoint, all data items are essentially ``data tasks" (shown as parallelograms in Figure \ref{fig:overall}) that can trigger the execution of their consumer tasks (shown as rectangles).

The second step, i.e. the focus of this paper, divides the PGT into a set of logical partitions such that certain performance requirements (e.g. total completion time, total data movement, etc.) are met under given constraints (e.g. resource footprint, collocation criteria, device locality, etc.). This step outputs the \pgt{} Partition (PGTP), which provides the Telescope Manager with an approximate solution to construct the observation scheduling blocks months or weeks prior to observation and compute resource allocation. An example of PGTP is shown at the bottom right of Figure \ref{fig:overall}, where 19 partitions are produced and one of them is visually expanded with 11 enclosing workflow tasks. Furthermore, a resource reservation that contains 19 nodes can be submitted to the telescope manager weeks before the associated observation takes place.

The third step maps each logical partition of the PG onto a given set of currently available resources in certain optimal ways. In principle, each partition is placed onto a physical compute node in the cluster. Such placement requires real-time information on resource availability, and we currently assume resource pools consisting of nodes with identical capabilities of computing, storage, and interconnect. In cases where the number of partitions \(p\) is greater than the number of available nodes \(m\), DALiuGE can be configured to merge the \(p\) PGT partitions into \(m\) virtual clusters with the goal of balancing the overall workload (both compute time and memory usage) evenly before mapping.

The final step involves optimal execution of tasks that have been allocated to a single node by the previous two steps. DALiuGE currently offloads this step to local schedulers provided by the host OS running on each compute node. We are currently working on the integration of graph-based GPU schedulers for dynamically scheduling GPU accelerated workflow tasks on single node with multiple GPUs.

In the following sections, we focus solely on the technical details of the second step --- dataflow partitioning.

\section{Dataflow partitioning}
During graph partitioning, a PGT of \(N\) vertices is decomposed into \(M\) partitions, each of which conceptually represents a compute node with a pre-defined resource capacity vector \(\boldsymbol{C}\). The goal of graph partitioning is to obtain an estimate on the minimum number \(M^*\) of compute nodes needed to execute the PGT and its corresponding PGT completion time \(T_{M^*}\). Initially, the partitioning algorithm lets \(M = N\) with each vertex being an individual partition. The algorithm then iteratively decreases \(M\) through partition merging (line 11 in Algorithm \ref{algo:dez}). This is equivalent to keeping the PGT completion time \(T\) monotonically non-increasing as exemplified in Figure \ref{fig:lpl_parts}. See \textbf{Theorem \ref{theory:lpl_time_decreasing}} for a proof. Therefore, a partition scheme that produces \(M^*\) ideally achieves the minimum PGT completion time \(T^*\), thus \(T_{M^*} = T^*\) under the current graph partitioning algorithm. This follows the ``data locality" principle which suggests that the unit cost of data movement between two partitions is far greater than that within the same partition. Therefore, fewer partitions lead to faster completion with less data movement, resource usage and lower operational cost.

On the other hand, a smaller \(M^*\) corresponds to a greater resource demand per partition since more Drops are allocated to each partition. This means the aggregated resource demand from concurrently-running Drops in a given partition is more likely to exceed \(\boldsymbol{C}\), slowing down the graph execution due to resource over-subscription. An ideal partitioning solution not only obtains an optimal \(M^*\) but also ensures that resource demands in all partitions stay below \(\boldsymbol{C}\) at any point during the graph execution. Satisfying this constraint avoids unpredictable execution delay due to resource over-subscription, thus ensuring \(T_{M^*} = T^*\). Formally, the graph partitioning is formulated as a constrained optimisation problem:

\begin{equation}
\begin{array}{rrclcl}
\displaystyle \min_{p} & \multicolumn{3}{l}{M(PGT, p)} \\
\textrm{s.t.} & R_i(t) & \leq & \boldsymbol{C}, & i = 1, \ldots, M., & \forall t \in \Big[0, T(PGT, p)\Big] \\
\end{array}
\label{eq:problem_def}
\end{equation}
where \(M(\cdot)\) is a function that outputs the number \(M\) of partitions given a \(PGT\) and a partition solution \(p\). \(T(\cdot)\) is a function that outputs the completion time \(T\) given a \(PGT\) and a partition solution \(p\). \(R_i(t)\) denotes the aggregated resource demand from all running Drops in partition \(i\) at time \(t\). We refer to the constraint defined in Equation \ref{eq:problem_def} as the \textbf{DoP constraint}, where ``DoP" stands for \textit{Degree of Parallelism}. Figure \ref{fig:res_dmd_per_part} exemplifies partitioning solutions that do (not) satisfy the DoP constraint.

\begin{figure}[h!]
\centering
\includegraphics[scale=0.17]{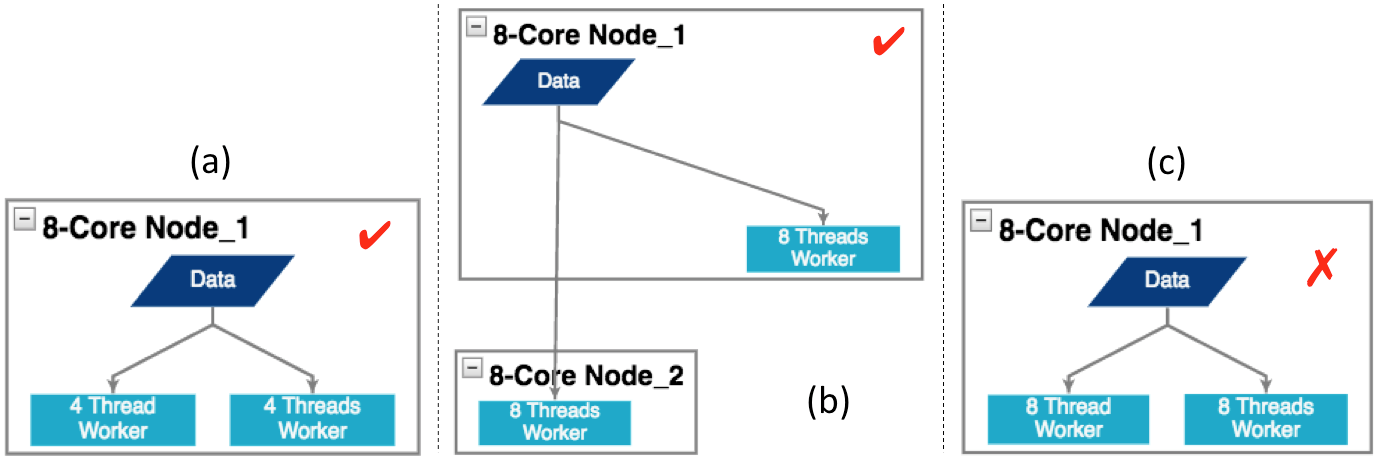}
\caption{Three solutions to partitioning a simple fork-like \pgt{}. Solution (a) places all three Drops inside a single partition. So the two worker Drops will run in parallel after the data becomes available, thus consuming 8 cores (four threads each) at the same time. This satisfies the DoP constraint given the resource capacity \(\boldsymbol{C}\) for a compute node includes 8 cores. However, solution (c) does not satisfy the DoP constraint since at some point 16 threads will be running in parallel on a single 8-Core machine. Consequently, the expected completion time for either worker is no longer guaranteed due to resource over-subscription. To remedy this, solution (b) separates the two worker Drops in two different partitions, each of which has sufficient resource capacity to execute 8 threads. Although the data movement between the two partitions incurs additional cost compared to Solution (c), Solution (b) produces far more reliable estimates on both completion time and resource demands with a potentially shorter completion time thanks to adequate resource provisioning.}
\label{fig:res_dmd_per_part}
\end{figure}

Once the optimal graph partitioning solution \(p\) is available, both \(M^*\) (known as the \textit{Physical Graph Template Partition}) and  \(T_{M^*}\) (i.e. \(T^*\)) are used by the telescope manager for the generation of observation and computing resource schedules well before the observation takes place. \newline

\subsection{Partitioning Algorithm}
The main idea of the partitioning algorithm (Algorithm \ref{algo:dez}) is to iteratively reduce data movement between inter-node Drops by ``merging" them into the same node, where the cost of intra-node communication is negligible. Given a PGT \(g\), the algorithm sorts all edges in \(g\) based on their weights in a descending order. The edge weight here denotes the volume of data ``on the move" from one Drop to the next. Each drop is initially allocated to a separate node. Then going through all edges in a descending order of their weights, the algorithm merges two partitions associated with the two Drops on both ends of the edge if the merged partition meets the DoP constraint defined in Equation \ref{eq:problem_def}. The algorithm is ``greedy" since it reduces larger costs before dealing with smaller ones. However, this may not necessarily lead to a globally optimal solution especially for large graphs. We are currently investigating various local search heuristics to overcome this limitation.

Although the iterative edge zeroing procedure is based on the graph clustering algorithm \cite{sarkar1987partitioning}, we added two important additional changes. First we allow two existing partitions to re-merge again in order to further reduce the number of partitions, which in turn reduces the total completion time as suggested in \textbf{Theorem \ref{theory:lpl_time_decreasing}}. Second, we evaluate the DoP constraint in order to accept or reject partition merging (line 12) proposals. The evaluation of the DoP constraint not only considers each graph vertex's processing requirement in terms of maximum number of concurrent threads, memory usage, etc., but also incorporates predefined resource capacities for each partition including number of cores, memory capacity, etc.

\IncMargin{1em}
\begin{algorithm}[h!]
\DontPrintSemicolon
\SetKwInOut{Input}{input}\SetKwInOut{Output}{output}
\SetKwData{Key}{key}
\SetKwData{Reverse}{reverse}
\SetKwData{Element}{element}
\SetKwData{Null}{is null}
\SetKwData{NotNull}{is not null}
\SetKw{kwAnd}{and}

 \Input{A DAG $g$ with a list $el$ of edges and a list $nl$ of nodes} 
 \Output{A list $l$ of partitions}

 \BlankLine
 initialise $l$ as an empty list\;
 $el.sort\_by\_weight(\Reverse \leftarrow true)$
 
 \ForEach{\Element $n$ of $nl$}{
    {\label{lt}
    $part\leftarrow create\_partition()$\;
    $part.add(n)$
    $l.add(part)$\;
    }
 }
 
 \ForEach{\Element $e$ of $el$}{
  $origin\_weight\leftarrow e.weight$\;
  $e.weight\leftarrow 0$\tcp*{edge zeroing}
    $u, v\leftarrow e.nodes()$\;
    $new\_part\leftarrow try\_merge\_partition(l, u, v)$
    
    \If(){$new\_part == NULL$}{
        $e.weight\leftarrow origin\_weight$
    }
   
 }
 
 \Return $l$
 
 \caption{The partitioning algorithm based on \cite{sarkar1987partitioning} with two important additions --- evaluation of the DoP constraint and merger of existing partitions}
 \label{algo:dez}
\end{algorithm}
\DecMargin{1em}

\begin{theorem}
\label{theory:lpl_time_decreasing}
The edge zeroing statement at line 10 in Algorithm \ref{algo:dez} ensures the completion time \(T\) of \(g\) is strictly non-increasing.
\end{theorem}
\begin{proof}
If the edge \(e\) is on the longest path \(L\) of \(g\) with a length \(T\), there are two possibilities after \(e\)'s weight becomes zero --- \(L\) remains the longest path of \(g\) or another path \(L'\) becomes the longest path of \(g\). In the first case, let \(T'\) be the new length of \(L\). It is easy to verify that \(T' == T - e.weight < T\). In the second case, let \(T'\) be the length of \(L'\). It must be true that \(T' \leq T\) because otherwise \(L'\) (rather than \(L\)) would have been the longest path before the edge zeroing takes place.

If the edge \(e\) is not on the longest path \(L\) of \(g\), there are also two possibilities after \(e\)'s weight becomes zero --- \(e\) remains off the longest path \(L\) of \(g\) or \(e\) becomes part of the ``new" longest path \(L'\) of \(g\). In the first case, since \(L\) is not affected whatsoever, its completion time \(T\) remains the same, thus non-increasing. In the second case, let \(T'\) be the length of \(L'\). It must be true that \(T' \leq T\) because otherwise \(L'\) (rather than \(L\)) would have been the longest path before the edge zeroing takes place.\end{proof}


\begin{figure*}
\centering
\includegraphics[scale=0.49]{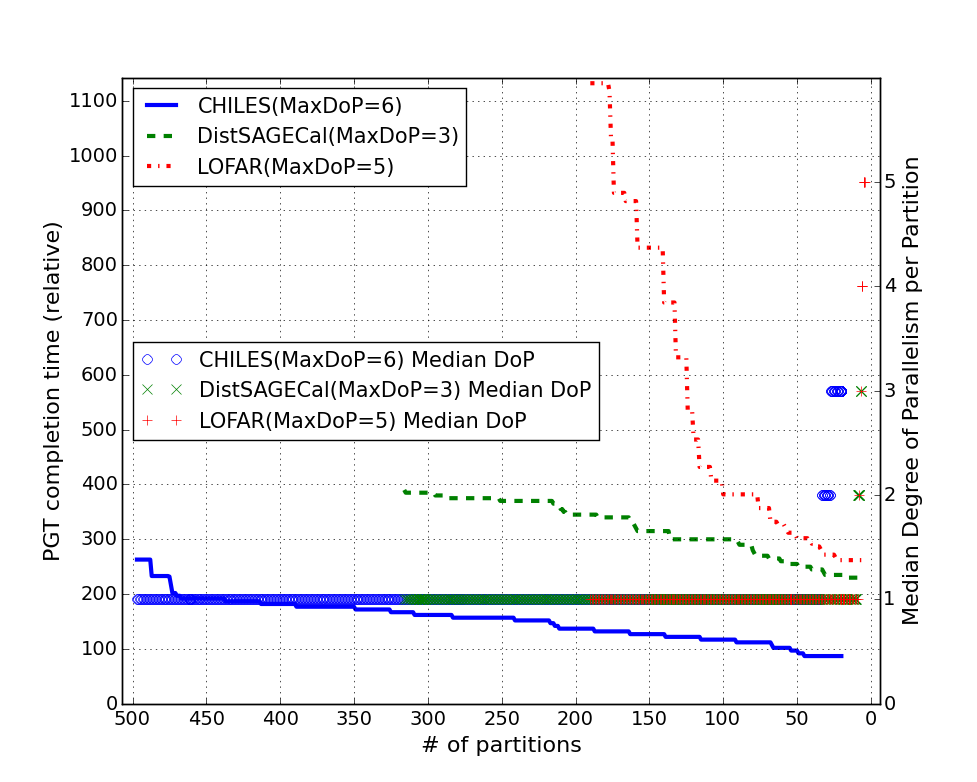}
\caption{The PGT completion time is monotonically non-increasing as the number of partitions decreases for three different radio astronomy pipeline graphs. It also shows the partition solution that produces the minimum number \(M^*\) of partitions (i.e. the bottom right end of each curve) also results in the shortest execution time \(T^*\).}
\label{fig:lpl_parts}
\end{figure*}

\subsection{DoP Constraint Evaluation}
In this subsection, we discuss the DoP evaluation algorithm defined in the \texttt{try\_merge\_partition} function called at line 12 in Algorithm \ref{algo:dez}. As shown in Equation \ref{eq:problem_def}, this boils down to efficiently computing the total resource usage \(R_i(t)\) summed over all running Drops inside a given partition \(i\) at a particular time \(t\). To do this, we first establish the equivalence between the set \(D(t)\) of Drops running in parallel at time \(t\) and the concept of \textit{antichain}~\cite{marcus2008graph} --- a set of \textbf{mutually unreachable} vertices of a DAG \(g\) associated with a given partition.

\begin{theorem}
\label{theory:antichain}
If all Drops in \(D(t)\) are running in a non-streaming mode,  \(D(t)\) is an antichain of \(g\).
\end{theorem}
\begin{proof}
The non-streaming running mode excludes one possible form of parallelism --- pipelining. All other forms of parallelisms require Drops in \(D(t)\) be mutually unreachable on \(r\) because otherwise they would never have been running in parallel due to their inter-dependencies as a result of reachability.
\end{proof}

We define the \textbf{length} \(L\) of an antichain \(D(t)\) as the number of Drops in \(D(t)\), and define the \textbf{weighted length} \(W\) of an antichain \(D(t)\) as the aggregated weight summed over all Drops in \(D(t)\). The weight of the \(j\)th Drop \(d_j\) in an antichain is the pre-determined peak resource usage denoted by \(w(d_j)\). Let \(\boldsymbol{A}\) denote the set of all antichains in an partition graph \(i\). It then follows from Theorem \ref{theory:antichain} that the total resource usage \(R_i(t)\) is bounded by some antichain(s) \(D\) that has the maximum (longest) weighted length amongst all antichains in \(\boldsymbol{A}\):
\begin{equation}
\begin{split}
R_i(t) \leq{} & W_{max} = \max_{D} \sum_{j=1}^{L} w(d_j), \\
&\text{where} \; d_j \in D, \; L = |D|, \; D \in \boldsymbol{A}, \; \forall t \in \Big[0, T(PGT, p)\Big]
\end{split}
\label{eq:maxantichain}
\end{equation}
Equation \ref{eq:maxantichain} bounds a time-dependent value \(R_i(t)\) by a time-invariant constant \(W_{max}\) such that if \(W_{max} \leq \boldsymbol{C}\) for a given partition, the constraint condition \(R_i(t) \leq \boldsymbol{C}\) in Equation \ref{eq:problem_def} will be satisfied. However, finding the antichain \(D^*\) that produces \(W_{max}\) is not trivial since the cardinality of \(\boldsymbol{A}\) --- the total number of antichains in a partition graph \(g\) --- can be in the order of \(2^n\), with \(n\) being the number of vertices in \(g\). Therefore, enumeration and evaluation of all antichains is computationally unfeasible in practice, where a typical partition has at least tens or even hundreds of tasks (e.g. there could be up to one billion antichains for a graph with merely 30 vertices).

To compute the maximum antichain length for a given graph in polynomial time, one can apply Dilworth's Theorem \cite{dilworth1950decomposition}, which states that the maximum length of an antichain is equal to the minimum number of chains needed to fully ``cover" the graph. In particular Fulkerson \cite{fulkerson1956note} established the equivalence between the maximum antichain length and the maximum matching in a constructed split graph (a.k.a. bipartite graph). As a result, the longest antichain --- the antichain that has the maximum cardinality --- of a graph can be discovered in \(O(|E|\sqrt{|V|})\) time.
However, Equation \ref{eq:maxantichain} suggests that the longest antichain does not necessarily have the longest weighted length unless \(w(d_j) = 1, \; \forall j \in [1, L]\). Hence, whilst we can efficiently solve \(W_{max}\) for a special case where each Drop consumes only one unit of resource (e.g. 1 core, 1G of RAM, etc.), we need a different algorithm to evaluate more generic cases where Drops consume arbitrary units of resources (e.g. 16 cores, 375 MB of RAM).


In the following, we discuss details of Algorithm \ref{algo:k-family} that efficiently computes \(W_{max}\) for generic cases based on Cong \cite{cong1993computing} to compute a maximum weighted \(k\)-family. While a \(k\)-family covers a union of at most \(k\) antichains in a DAG, we are interested only in a special case (where \(k = 1\)) in order to solve our problem of computing the maximum weighted length of a single antichain.

\IncMargin{1em}
\begin{algorithm}[h!]
\DontPrintSemicolon
\SetKwInOut{Input}{input}
\SetKwInOut{Output}{output}
\SetKw{kwAnd}{and}
\SetKwData{Key}{key}
\SetKwData{Reverse}{reverse}
\SetKwData{Element}{element}
\SetKwData{Null}{is null}
\SetKwData{null}{null}
\SetKwData{NotNull}{is not null}
\SetKwData{This}{this}
\SetKwData{Boolean}{boolean}
\SetKwData{Capacity}{capacity}
\SetKwData{Cost}{cost}
\SetKwData{Infinity}{infinity}
\SetKwProg{Fn}{function}{:}{end}

 \Input{partitions $A$ and $B$ with their associated DAGs $g_A$ and $g_B$ 
 \\ an optional $g\_dag$ representing the unpartitioned physical graph template}
 \Output{the maximum weighted antichain length of the merged partition $A \bigcup B$}
 
 \BlankLine

\Fn{get\_pi\_solution(g)}{
    $S\leftarrow \text{create\_split\_graph}(g)$\;
    $H\leftarrow \text{admissible\_graph}(S)$\;
    $f'\leftarrow \text{maximum\_flow}(H, s, t)$\;
    $R\leftarrow \text{residual\_graph}(H, f')$\;
    \ForEach{\Element $r\_node \in R.$\text{nodes}$()$}{
     
     \lIf{$R.\text{has\_path}(s, r\_node)$}{
        $pi[r\_node]\leftarrow 0$
     }\lElse{
        $pi[r\_node]\leftarrow 1$
     }
    }
    \Return $pi$
}
 
 
 
 $pi\leftarrow \text{get\_pi\_solution}(g_A \bigcup g_B)$\;
 $W_{max}\leftarrow 0$\;
 \For{$h\gets0$ \KwTo $1$}{
    \ForEach{\Element $nd_x \in S.X$}{
            $nd_y\leftarrow S.counter\_part(nd_x)$\;
            \If{$h = 1 - pi[nd_x] + pi[node\_S]$ \kwAnd $1 = pi[nd_y] - pi[nd_x]$}{
                $W_{max}\leftarrow W_{max} + S.\text{edge}(node\_S, nd_x).\Capacity$
            }
        
    }
 }
 \Return $W_{max}$
 \caption{Calculate \(W_{max}\), the maximum weighted antichain length in a partition}
 \label{algo:k-family}
\end{algorithm}
\DecMargin{1em}

The central idea of Algorithm \ref{algo:k-family} is to exploit the equivalence between the weighted maximum anti-chain of the original DAG \(g\) and the minimum-cost maximum-flow (MCMF) solution of the split graph \(S\) created at Line 2. The equivalence is proved in \cite{cong1993computing} and more generally in~\cite{cameron1985antichain}. Note that number of nodes of \(S\) is \(2V + 2\) where \(V\) is the number of the original DAG, which is the union \(g\) of the two DAGs \(g_A\) and \(g_B\). This ensures that a polynomial algorithm on \(S\) remains tractable on \(g\). 

To find the MCMF solution, we first derive the admissible graph \(H\) from \(S\) (line 3), and run the normal maximum flow algorithm~\cite{goldberg1988new} to obtain the flow \(f'\) in \(O(V^2\sqrt{E})\) time (line 4). We then construct the residual graph \(R\) from \(f'\) (line 5). \(R\) has the identical set of vertices as \(H\), and if there are no edges going from the source vertex \(s\) of \(R\) to some vertex \(x\), then we set the node potential \(\pi\) of \(x\) to 1 (line 8). In the end, the maximum weighted antichain \(W_{max}\) is calculated (line 14 to 20) based on expressions defined in Theorem 3.1 \cite{cong1993computing}. Figure \ref{fig:lpl_parts} shows the results of running Algorithm \ref{algo:dez} and \ref{algo:k-family} by scheduling three different radio interferometry imaging workflows.

\section{Conclusions}
Optimal scheduling of large-scale, data-intensive workflows is challenging. In this paper, we discussed related work on graph scheduling and proposed polynomial time optimization methods that minimize both workflow execution time and resource footprint while meeting resource demand constraints imposed by individual algorithms. We show preliminary results obtained from three radio astronomy data pipelines.


%
\bibliographystyle{abbrv}
\bibliography{sig-alternate-sample}  

\begin{thebibliography}{10}

\bibitem{bateni2017affinity}
M.~Bateni, S.~Behnezhad, M.~Derakhshan, M.~Hajiaghayi, R.~Kiveris, S.~Lattanzi,
  and V.~Mirrokni.
\newblock Affinity clustering: Hierarchical clustering at scale.
\newblock In {\em Advances in Neural Information Processing Systems}, pages
  6867--6877, 2017.

\bibitem{bokhari1981shortest}
S.~H. Bokhari.
\newblock A shortest tree algorithm for optimal assignments across space and
  time in a distributed processor system.
\newblock {\em IEEE transactions on Software Engineering}, (6):583--589, 1981.

\bibitem{braun2015advancing}
R.~Braun, T.~Bourke, J.~Green, E.~Keane, and J.~Wagg.
\newblock Advancing astrophysics with the square kilometre array.
\newblock {\em Advancing Astrophysics with the Square Kilometre Array
  (AASKA14)}, 1:174, 2015.

\bibitem{cameron1985antichain}
K.~Cameron.
\newblock Antichain sequences.
\newblock {\em Order}, 2(3):249--255, 1985.

\bibitem{chaudhary1993generalized}
V.~Chaudhary and J.~K. Aggarwal.
\newblock A generalized scheme for mapping parallel algorithms.
\newblock {\em IEEE Transactions on Parallel and Distributed Systems},
  4(3):328--346, 1993.

\bibitem{cong1993computing}
J.~Cong.
\newblock {\em Computing maximum weighted k-families and k-cofamilies in
  partially ordered sets}.
\newblock Computer Science Department, University of California, 1993.

\bibitem{dennis1975preliminary}
J.~B. Dennis and D.~P. Misunas.
\newblock A preliminary architecture for a basic data-flow processor.
\newblock In {\em ACM SIGARCH Computer Architecture News}, volume~3, pages
  126--132. ACM, 1975.

\bibitem{dilworth1950decomposition}
R.~P. Dilworth.
\newblock A decomposition theorem for partially ordered sets.
\newblock {\em Annals of Mathematics}, pages 161--166, 1950.

\bibitem{fulkerson1956note}
D.~R. Fulkerson.
\newblock Note on dilworth’s decomposition theorem for partially ordered
  sets.
\newblock In {\em Proc. Amer. Math. Soc}, volume~7, pages 701--702, 1956.

\bibitem{goldberg1988new}
A.~V. Goldberg and R.~E. Tarjan.
\newblock A new approach to the maximum-flow problem.
\newblock {\em Journal of the ACM (JACM)}, 35(4):921--940, 1988.

\bibitem{karypis1998multilevelk}
G.~Karypis and V.~Kumar.
\newblock Multilevelk-way partitioning scheme for irregular graphs.
\newblock {\em Journal of Parallel and Distributed computing}, 48(1):96--129,
  1998.

\bibitem{kwok1999static}
Y.-K. Kwok and I.~Ahmad.
\newblock Static scheduling algorithms for allocating directed task graphs to
  multiprocessors.
\newblock {\em ACM Computing Surveys (CSUR)}, 31(4):406--471, 1999.

\bibitem{liou1997comparison}
J.-C. Liou and M.~A. Palis.
\newblock A comparison of general approaches to multiprocessor scheduling.
\newblock In {\em Parallel Processing Symposium, 1997. Proceedings., 11th
  International}, pages 152--156. IEEE, 1997.

\bibitem{marcus2008graph}
D.~Marcus.
\newblock {\em Graph theory: a problem oriented approach}.
\newblock The Mathematical Association of America, 2008.

\bibitem{martella2017spinner}
C.~Martella, D.~Logothetis, A.~Loukas, and G.~Siganos.
\newblock Spinner: Scalable graph partitioning in the cloud.
\newblock In {\em Data Engineering (ICDE), 2017 IEEE 33rd International
  Conference on}, pages 1083--1094. Ieee, 2017.

\bibitem{sarkar1987partitioning}
V.~Sarkar.
\newblock {\em Partitioning and scheduling parallel programs for execution on
  multiprocessors}.
\newblock PhD thesis, 1987.

\bibitem{stone1977multiprocessor}
H.~S. Stone.
\newblock Multiprocessor scheduling with the aid of network flow algorithms.
\newblock {\em IEEE transactions on Software Engineering}, (1):85--93, 1977.

\bibitem{topcuoglu2002performance}
H.~Topcuoglu, S.~Hariri, and M.-y. Wu.
\newblock Performance-effective and low-complexity task scheduling for
  heterogeneous computing.
\newblock {\em Parallel and Distributed Systems, IEEE Transactions on},
  13(3):260--274, 2002.

\bibitem{towsley1986allocating}
D.~Towsley.
\newblock Allocating programs containing branches and loops within a multiple
  processor system.
\newblock {\em IEEE Transactions on Software Engineering}, (10):1018--1024,
  1986.

\bibitem{tsourakakis2014fennel}
C.~Tsourakakis, C.~Gkantsidis, B.~Radunovic, and M.~Vojnovic.
\newblock Fennel: Streaming graph partitioning for massive scale graphs.
\newblock In {\em Proceedings of the 7th ACM international conference on Web
  search and data mining}, pages 333--342. ACM, 2014.

\bibitem{wu2017daliuge}
C.~Wu, R.~Tobar, K.~Vinsen, A.~Wicenec, D.~Pallot, B.~Lao, R.~Wang, T.~An,
  M.~Boulton, I.~Cooper, et~al.
\newblock Daliuge: A graph execution framework for harnessing the astronomical
  data deluge.
\newblock {\em Astronomy and Computing}, 20:1--15, 2017.

\end{thebibliography}
%
%

\end{document}